\newcommand{\version}{Dezember, 2013}     

\documentclass[12pt]{article}
\usepackage{bbm} 
\usepackage{mathrsfs} 
\usepackage{a4,amsthm,amsfonts,latexsym,amssymb}
\usepackage{mdwlist}

\swapnumbers 

\theoremstyle{plain}
\newtheorem{thm}{THEOREM}
\newtheorem{cor}[thm]{COROLLARY}
\newtheorem{lem}[thm]{LEMMA}
\newtheorem{define}[thm]{DEFINITION}
\newtheorem{proposition}[thm]{PROPOSITION}

\newcommand{\beq}{\begin{equation}}
\newcommand{\eeq}{\end{equation}}
\def\beqa{\begin{eqnarray}}
\def\eeqa{\end{eqnarray}}

\newcommand{\Tr}{{\rm Tr}}

\newcommand{\lspan}{{\rm lin}}

\newcommand{\C}{{\mathbb C}}

\newcommand{\R}{{\mathbb R}}

\newcommand{\one}{{\mathbbm 1}}
\newcommand{\Tf}{{\mathfrak{T}}}

\newcommand{\M}{{\mathcal M}}

\newcommand{\Q}{{\mathcal Q}}

\newcommand{\F}{{\mathcal F}}

\newcommand{\Tt}{{\mathcal T}}
\renewcommand{\L}{{\mathcal L}}

\newcommand{\eps}{\varepsilon}

\newcommand{\BB}{{\mathscr B}}
\newcommand{\HH}{{\mathscr H}}
\newcommand{\DD}{{\mathscr D}}
\newcommand{\CC}{{\mathscr C}}

\newcommand{\LL}{{\mathscr L}}

\newcommand{\tdt}{\!\cdot\!}
\newcommand{\half}{\mbox{$\frac{1}{2}$}} 

\def\bce{\begin{center}}
\def\ece{\end{center}}
\def\bit{\begin{itemize}}
\def\eit{\end{itemize}}

\date{\small\version}

\begin{document}
\markboth{\scriptsize{ Generators of Quantum Dynamical Semigroups \qquad  \version}}
         {\scriptsize{ Generators of Quantum Dynamical Semigroups \qquad  \version}}

\title{
\vspace{-80pt}
\bf{Generalized Kraus Operators and Generators of Quantum Dynamical Semigroups}}
\author{\vspace{30pt} Sabina Alazzawi$^1$, Bernhard Baumgartner$^2$
\\
\vspace{-4pt}
\small{Fakult\"at f\"ur Physik, Universit\"at Wien}\\
\small{Boltzmanngasse 5, A-1090 Vienna, Austria}}

\maketitle


\begin{abstract}
Quantum dynamical semigroups play an important role in the description of physical processes such as diffusion, radiative decay or other non-equilibrium events. Taking strongly continuous and trace preserving semigroups into consideration, we show that,
under a special criterion,
the generator of such a group admits a certain generalized standard form, thereby shedding new light on known approaches in this direction. Furthermore, we illustrate our analysis in concrete examples.
\\[20ex]
\\[3ex]

\end{abstract}

\footnotetext[1]{\texttt{sabina.alazzawi@univie.ac.at}}
\footnotetext[2]{\texttt{Bernhard.Baumgartner@univie.ac.at}}

\maketitle


\newpage

\section{Introduction}\label{intro}

Time evolutions of open quantum systems are not appearing as unitary maps.
Phenomenological equations describing non-equilibrium processes,
like diffusion, radiative decay, radiative excitation, etc., have no
terms describing memory-effects, they are Markovian.
These equations show the arrow of time, one cannot follow evolutions indefinitely backwards,
so they are to be described with \textbf{semigroups} of maps $\Tt^t$, acting on density matrices.

Basic dynamics is reversible.
In the quest for answering the question of
how irreversibility can emerge one has to perform limiting processes.
The methods of deriving $\Tt^t$ from basic unitary dynamics, see \cite{D74,D76,D76b}, using weak coupling or singular coupling limits,
give rise to semigroups of \textbf{completely positive} maps.

The generally possible structure of generators of such quantum dynamical processes
has been completely characterized for norm-continuous  semigroups,
yielding the ``GKS-Lindblad-equation'', \cite{GKS76,L76}.
Further attempts to study strongly continuous evolutions by E.B.Davies, \cite{D77,D77b,D80},
led, with some special assumptions, to a ``standard form'' of generators.
A.S.Holevo extended this study further, \cite{K95,H96,H97}.
He defines a generator in the Heisenberg picture as a quadratic form on the Hilbert space.
Crucial results of investigations in this sense
are presented by F.Fagnola and R.Rebolledo   in \cite{FR06}, sect. 3.

\par
In this paper we rework their methods. We
define a \textbf{generalized standard form}
for each generator of a QDS (quantum dynamical semigroup which is strongly continuous)
in the Schr\"{o}dinger picture and from a purely functional analytic point of view.
It allows for a complete characterization in case
the QDS is unambiguously defined by its ``matrix normal'' action on the matrix units
in some basis of the Hilbert space.
Two new examples of a QDS are presented, showing all the special features of the present analysis.
One is as simple and clear as possible, the other one models a natural process
observed in modern quantum physics.

\section{Setup and main statements}\label{setup}

We use the known fact that each strongly continuous semigroup on a Banach space
has a generator $\L$,
a closed operator with dense domain of definition $\DD(\L)$ which is
the set of all those Banach space elements $\rho$ for which the limit $t\rightarrow 0$
in the  formula defining $\L$,
\beq\label{generator}
\L(\rho)= \lim_{t\rightarrow 0}\frac{1}{t} \left( \Tt^t(\rho)- \rho \right),
\eeq
exists in norm. (In \cite{D76b} such a semigroup is denoted as  ``$c_0$ semigroup''.)
Now a QDS acts on three levels:  Hilbert space, density matrices, ``super operators'' acting on density matrices.
The goal is to represent the generator $\L$ by way of decomposing its action
into actions described by Hilbert space operators. Under certain assumptions stated below, this can be done by considering
\begin{itemize}
  \item  a separable Hilbert space $\HH$,
  \item  the Banach space of trace class operators $\mathfrak{T}(\HH)$ with the trace norm,
  \item  a strongly continuous semigroup of trace preserving completely positive
  maps $\Tt^t :$ $\mathfrak{T}(\HH)\rightarrow\mathfrak{T}(\HH)$ with generator $\L$.
\end{itemize}

\begin{define}{ \textbf{\textrm{ Generalized standard form}}}\label{gksldh}
\newline
Consider $L_k$ and $M$, linear operators on $\HH$ with dense domains of definition, and with the properties that
$M$ is  closed and generates a semigroup on $\HH$, and each $L_k$ is relatively bounded with respect to $M$.
Consider, moreover, ``generalized operators''
$L_k^\dag$ and $M^\dag$, linear operators mapping $\HH$ into the conjugate algebraic dual of $\DD(M)$,
such that for all $\psi\in\HH$ and for all $\eta\in\DD(M)$,
\beq\label{genop1}                                                                   
\langle \eta|L_k^\dag|\psi\rangle := \langle \psi|L_k|\eta\rangle^\ast,  \quad \quad
\langle \eta|M^\dag|\psi\rangle   := \langle \psi|M|\eta\rangle^\ast.
\eeq
Then the generator $\L$ of a trace preserving QDS on $\mathfrak{T}(\HH)$ is in generalized standard form if
\beq\label{gensform}
\L(\rho) = \sum_k L_k\tdt \rho\tdt L_k^\dag - M\tdt\rho - \rho\tdt M^\dag,
\eeq
\beq
  \sum_k L_k^\dag L_k = M+M^\dag. \label{formequal}
\eeq
Equation (\ref{gensform}) holds for  $\rho$ of finite rank
which can be represented as finite sums of dyadic products
$\sum_n |\phi_n\rangle\langle\psi_n|$ with $\{\phi_n,\psi_n\}\subset\DD(M)$.
The domain $\DD(\L)$ is the closure of this set $\DD_0$
in graph norm $\|\rho\|_\L=\|\rho\|_1 +\|\L(\rho)\|_1$.
Equation (\ref{formequal}) holds for the noted entities as quadratic forms with $\DD(M)$ as the common form domain.
\end{define}
In sloppy words, the operators $M^\dag$ and $L_k^\dag$ have their domain of definition at the left hand side.
In Section \ref{discussion} we discuss these notations and the relation to Davies' investigations.

To prove the existence of such a generalized standard form we
need the assumption that there is some basis,
in which the matrix elements of density operators are differentiable functions of time.
We, therefore, consider:
\begin{itemize}
\item  a basis $\{e_k\}\subset\HH$ and the dense set  of finite linear combinations of these basis elements
  \beq\label{de}
  \DD_e:=\lspan\{e_k\},
  \eeq
  \item  the set of finite linear combinations of matrix units in this basis
  \beq\label{desquared}
  \DD_e^2 := \lspan\{|e_k\rangle\langle e_\ell|\}\subset\mathfrak{T}(\HH).
  \eeq
\end{itemize}
\begin{define}{ \textrm{\textbf{ Matrix normal QDS}}}\label{decomposable}
\newline
A QDS on $\mathfrak{T}(\HH)$ is called
``matrix normal'' if there exists a basis $\{e_k\}$ of $\HH$ such that $\DD_e^2$ as defined in  (\ref{desquared})
is a core for its generator $\L$.
\end{define}
The operator $\L$ is then completely defined by its action on each $\rho$ in the core $\DD_e^2$.
It is hence \emph{decomposable} into its actions on the matrix units $|e_k\rangle\langle e_\ell|$ in the basis $\{e_k\}$.
It is, for our intentions, not sufficient to assume nothing but that $\DD_e^2$ is just contained in the domain of definition of $\L$.
In Section \ref{discussion} we explain this remark and give an example.

This formal assumption on decomposability with respect to matrix elements
is now related to the description of generators $\L$ in the generalized standard form,
using operators on the underlying Hilbert space $\HH$,
without a priori reference to a basis.
The main result of the present investigation is:
\begin{thm}{ \textbf{\textrm{ Matrix normality is equivalent to the existence of a generalized standard form.}}}\label{gksldhthm}
\newline
There exists a generalized standard form for the generator $\L$
of a trace preserving QDS if and only if it is matrix normal.
In this case  $\DD_e\subset\DD(M)$.
\end{thm}
%
%


\section{Proofs and further details}\label{procedure}

\subsection{From decomposability to the standard form.}\label{adctostandard}

Before we come to the proof of Theorem \ref{gksldhthm}, some preparation is necessary. We start by stating the following
\begin{lem}\label{lprop}

The generator $\L$ of a matrix normal trace preserving QDS has the properties:

\begin{enumerate}
  \item $\forall\phi\in\DD_e: \quad \langle\phi\,|\L(|\phi\rangle\langle\phi|)\,|\phi\rangle\leq 0$, \label{litem1}
  \item $\forall\phi\in\DD_e,\forall\psi\textrm{ with } \langle\phi|\psi\rangle=0 :\quad
           \langle\psi\,|\L(|\phi\rangle\langle\phi|)\,|\psi\rangle\geq 0$,       \label{litem2}
  \item $\forall\{\phi_k\in\DD_e,\psi_k\} \textrm{ with } \sum_k\langle\phi_k|\psi_k\rangle=0 :\quad
           \sum_{k,\ell}\langle\psi_k|\,\L(|\phi_k\rangle\langle\phi_\ell|)\,|\psi_\ell\rangle\geq 0$. \label{litem3}
\end{enumerate}
\end{lem}
\begin{proof}
Item (\ref{litem1}) is a consequence of preservation of both positivity and trace by transformation of $|\phi\rangle\langle\phi|$ by $\Tt^t$.
Starting at $t=0$ with $\rho=|\phi\rangle\langle\phi|$,
some new positive diagonal elements  may appear in $\Tt^t(|\phi\rangle\langle\phi|)$.
Since  the sum of all the diagonal elements is constant, the original single element,
which was ``$1$'', when $\|\phi\|=1$, has to diminish.
Item (\ref{litem2}) concerns the newly appearing diagonal elements as $t$ increases,
starting from zero, and follows again from positivity of $\Tt^t$.
Item (\ref{litem3}) follows from complete positivity of $\Tt^t$. It is item (\ref{litem2})
      applied to the generator  $\hat{\L}$ of the semigroup  $\hat{\Tt^t} =  \Tt^t \otimes \textrm{Id}$,
      acting on density matrices on $\HH\bigotimes\ell^2$,
      where $\textrm{Id}$ is the identity-map,
      with $\hat{\phi}=\sum_k \phi_k\otimes f_k$,  $\hat{\psi}=\sum_k \psi_k\otimes f_k$,
      where $\{f_k\}$ is an orthonormal basis of $\ell^2$.
\end{proof}
The subsequent construction follows Holevo \cite{H97};
(see also, for the case of norm continuous semigroups, \cite{F99};
the method, according to F.Fagnola, has already been used earlier by A.Frigerio):

\begin{define}{ \textbf{The generator of contraction in Hilbert space, M}}\label{mcontractor}
\newline
For a matrix normal QDS
choose some vector $\chi\in\DD_e$ with $\|\chi\|=1$.
On $\DD(M):=\{\psi\,:\,|\psi\rangle\langle\chi|\in \DD(\L) \}$ an operator $M$
is defined by
\beq\label{em}
M\psi := -\L(|\psi\rangle\langle\chi|)\chi +
\half \langle\chi|\,\L(|\chi\rangle\langle\chi|)\,|\chi\rangle\, \psi.
\eeq
\end{define}
The naming is explained in Proposition \ref{contractor}.
Note that $\DD_e\subset \DD(M)$, which will be used in Equation (\ref{expandermap}).

\begin{proposition}{ \textbf{The expander in the generator of the QDS}}\label{expander}
\newline
For a matrix normal QDS
the map
\beq\label{expandermap}
\L_+ :\,\,\rho\mapsto \L(\rho)+M\rho+\rho M^\dag
\eeq
is a completely positive map from $\DD_e^2$ into $\mathfrak{T}(\HH)$.
\end{proposition}
\begin{proof}
Each $\sigma\in\DD_e^2$  can be expanded as a linear combination of at most four positive matrices of finite rank.
Each positive matrix $\rho\in\DD_e^2$  can be represented as $\rho=\sum_{n=1}^N  |e_n\rangle\langle e_n|$
for some elements $e_n\in\DD_e$, and there exists
$$
\Tr\{M\rho+\rho M^\dag\} = \sum_{n=1}^N \langle e_n|(M+M^\dag)|e_n\rangle <\infty.
$$
Now $\L$ generates a semigroup which preserves the trace, so for all $\rho\in\DD(\L)$
\beq\label{ltrace}
\Tr\{\L(\rho)\}=\left[\frac d {dt}\Tr\{\Tt^t(\rho)\}\right]_{t=0}=0,
\eeq
and this implies, for all $\rho\in\DD_e^2$
\beq
\Tr\{\L_+(\rho)\}=\Tr\{M\rho + \rho M^\dag\},  \label{lplusone}
\eeq
and $\L_+$ is well defined on $\DD_e^2$.


The criterion for complete positivity states that
$\sum_{k,\ell} \langle\psi_k|\,\L_+(|\phi_k\rangle\langle\phi_\ell|)\,|\psi_\ell\rangle$
is non-negative for all
$
N$ and $ \{(\phi_k,\psi_k), k=\,1 \ldots N,\,\,\phi_k\in\DD_e\}.
$
To check whether this criterion is fulfilled,
we insert (\ref{em}) into (\ref{expandermap}). 
We get
\beqa
\langle\psi_k|\,\L_+(|\phi_k\rangle\langle\phi_\ell|)\,|\psi_\ell\rangle = & & \nonumber\\
 \langle\psi_k|\,\L(|\phi_k\rangle\langle\phi_\ell|)\,|\psi_\ell\rangle
&+& \langle\chi|\,\L(|\chi\rangle\langle\chi|)\,|\chi\rangle \cdot \langle\psi_k|\phi_k\rangle\cdot\langle\phi_\ell|\psi_\ell\rangle\nonumber\\
- \langle\psi_k|\,\L(|\phi_k\rangle\langle\chi|)\,|\chi\rangle \cdot\langle\phi_\ell|\psi_\ell\rangle
&-& \langle\chi|\,\L(|\chi\rangle\langle\phi_\ell|)\,|\psi_\ell\rangle \cdot \langle\psi_k|\phi_k\rangle. \label{hol}
\eeqa
Defining
\beq\label{double}
\phi_{N+k}:=- \langle\psi_k|\phi_k\rangle\cdot \chi,\,\qquad \psi_{N+k}:=\chi,
\eeq
we continue the reformulation of (\ref{hol}), writing it as
\beq
\sum_{i=0}^1 \,\sum_{j=0}^1\,\, \langle\psi_{k+iN}|\,\L(|\phi_{k+iN}\rangle\langle\phi_{\ell+jN}|)\,|\psi_{\ell+jN}\rangle .
\eeq
Summation over all indices gives
\beq
\sum_{k=1}^N \sum_{\ell=1}^N    \langle\psi_k|\,\L_+(|\phi_k\rangle\langle\phi_\ell|)\,|\psi_\ell\rangle
= \sum_{k=1}^{2N} \sum_{\ell=1}^{2N}    \langle\psi_k|\,\L(|\phi_k\rangle\langle\phi_\ell|)\,|\psi_\ell\rangle \geq 0.
\eeq
The inequality at the end holds because of  the property stated in Lemma (\ref{lprop}.\ref{litem3}) of $\L$
and $\sum_{k=1}^{2N}\langle\phi_k|\psi_k\rangle = 0$, which follows from (\ref{double}).
Complete positivity of $\L_+$ is thus established.
\end{proof}

With these preparations we now come to the proof of the algebraic part of Theorem \ref{gksldhthm}:
\begin{proof}{ \textbf{Algebraic skeleton of the extended standard form:}}
\newline
Due to Proposition \ref{expander}, Theorem \ref{krausthm} and Corollary \ref{krausdomain},
given in the Appendix in Section \ref{unboundedcp},
we have for all $ \rho \in \DD_e^2$ that $
\L(\rho)$ can be decomposed as
$$
\L(\rho)=\L_+(\rho)-M\rho-\rho M^\dag.
$$
Moreover, there exist generalized Kraus Operators $L_k$, acting on $\DD_e$, such that
$$
\L_+(\rho)=\sum_k L_k\tdt\rho\tdt L_k^\dag.
$$

Proof of Equation (\ref{formequal}):
\newline
Operations which are standard procedures in dealing with ordinary operators
have to be carefully adapted to the case where generalized operators appear.
Consider $\rho=|\phi\rangle\langle\phi| \in \DD_e^2$ with $\|\phi\|=1$.
So $\L(\rho)\in\mathfrak{T}(\HH)$ and $\Tt^t(\rho)$ is differentiable:
$$
0= \frac d{dt} \Tr\left(\Tt^t(\rho)\right)|_{t=0} = \Tr\left(\L(\rho)\right)=
\sum_\alpha\langle\alpha|\left(\sum_k L_k\rho L_k^\dag -M\rho-\rho M^\dag\right)|\alpha\rangle=
$$
$$
=\sum_\alpha \left( \sum_k\langle\alpha|L_k|\phi\rangle\langle\phi|L_k^\dag|\alpha\rangle
-\langle\alpha|M|\phi\rangle\langle\phi|\alpha\rangle- \langle\alpha|\phi\rangle\langle\phi|M^\dag|\alpha\rangle\right),
$$
for any basis $\alpha$. We choose a basis with $|\alpha=1\rangle=|\phi\rangle$,
so $\langle\alpha|\phi\rangle=\delta_{\alpha,1}$, and
$$
\Tr\left(\L(\rho)\right)=\sum_\alpha\left( \sum_k |\langle\alpha|L_k|\phi\rangle|^2
-(\langle\alpha|M|\phi\rangle\ + \langle\phi|M^\dag|\alpha\rangle)\delta_{\alpha,1} \right)=
$$
$$
=\sum_\alpha \left(\sum_k |\langle\alpha|L_k|\phi\rangle|^2 \right)  -\langle\phi|(M + M^\dag)|\phi\rangle.
$$
Since $\Tr\left(\L(\rho)\right)=0$ and the expectation value of $M$ is finite, the first term with the double sum is also finite. This sum over non-negative numbers, however, can be rearranged as follows:
$$
\sum_\alpha \sum_k |\langle\alpha|L_k|\phi\rangle|^2 =
\sum_k \sum_\alpha \langle\phi| L_k^\dag  |\alpha\rangle \langle\alpha|  L_k|\phi\rangle =
 \langle\phi|\sum_k L_k^\dag L_k|\phi\rangle.
$$
Therefore, the formula
$$
\langle\phi|\left(\sum_k L_k^\dag L_k -M-M^\dag  \right)|\phi\rangle=0
$$
holds for each vector $\phi\in\DD_e$, and Equation (\ref{formequal}) is proven
to be valid on the restricted form domain $\DD_e$.
\end{proof}

Having built the algebraic structure, we consider analytic-topological features.
From Equation (\ref{formequal})
on its skeleton domain it follows immediately that each $L_k$ is bounded relative to $M$, that is
$$
\|L_k\phi\|^2=2 \textrm{Re} \langle\phi|M|\phi\rangle \leq 2\|\phi\|\cdot\|M\phi\| < (\|\phi\|+\|M\phi\|)^2.
$$
The $L_k$ can therefore be uniquely extended to the entire domain $\DD(M)$. The super operator $\L_+$ and the validity of Equation (\ref{gensform}) can be extended to $\DD_0$,
and  Equation (\ref{formequal}) holds on the whole form domain.

In order to finish the proof of one half of Theorem \ref{gksldhthm}, our last task is to show
that $M$ is closed and generates a contraction semigroup.
\begin{proposition}{ \textbf{ M is closed.}}
\newline
The operator $M$, the mover-and-contractor presented in Definition \ref{mcontractor}, is closed.
\end{proposition}
\begin{proof}

We consider a sequence $\psi_n \in \DD(M)$, with $\psi_n\rightarrow\psi$ and $M\psi_n\rightarrow\eta$.
With $\psi_{m,n}:=\psi_m-\psi_n$,
using   Equation (\ref{lplusone}), we have
$$
0\leq\Tr\{\L_+(|\psi_{m,n}\rangle\langle \psi_{m,n}|)\}=
\Tr\{M|\psi_{m,n}\rangle\langle\psi_{m,n}|  + |\psi_{m,n}\rangle\langle\psi_{m,n}|M^\dag\}$$
$$\leq 2\cdot\|\psi_m-\psi_n\|\cdot\|M(\psi_m-\psi_n)\| \rightarrow 0
$$
as $m,n\,\rightarrow\infty$.
All following limits of operators are now existing in trace norm.

For any $\gamma\in\DD_e$, any $z\in\C$, we have
$$
0\leq\L_+(|\gamma+z \psi_{m,n}\rangle\langle\gamma+z \psi_{m,n}|)=$$
$$=\L_+(|\gamma\rangle\langle\gamma|)+z\L_+(|\psi_{m,n}\rangle\langle\gamma|)+
z^\ast\L_+(|\gamma\rangle\langle\psi_{m,n}|)+|z|^2\L_+(|\psi_{m,n}\rangle\langle\psi_{m,n}|).
$$

The first term is constant and the last one vanishes in the limit $m,n\rightarrow\infty$ for arbitrary $|z|$.
The two terms in between, depending linearly on $z$ and being negative for some argument of $z$, have to remain bounded
in relation to $\L_+(|\gamma\rangle\langle\gamma|)$, independently of $z$, hence each one
has to go to zero, especially $\L_+(|\psi_{m,n}\rangle\langle\gamma|)\rightarrow 0$.

It follows that
$\L_+(|\psi_n\rangle\langle\gamma|)$ is a Cauchy sequence, and therefore also
$\L(|\psi_n\rangle\langle\gamma|)$ converges. Since $\L$ is closed,  $|\psi\rangle\langle\gamma|\in\DD(\L)$
for each $\gamma\in\DD_e$
and hence $\psi\in\DD(M)$ and $M\psi=\eta$, by definition.
\end{proof}

\begin{proposition}{ \textbf{ M generates a semigroup on the Hilbert space.}}\label{contractor}
\newline
The operator $M$ is maximal accretive.
It generates the contractive semigroup $e^{-tM}$ on the Hilbert space $\HH$.
\end{proposition}

\begin{proof}
$M$ is accretive ($-M$ is dissipative), i.e. $\textrm{Re}\langle\phi|M|\phi\rangle\geq 0$
$\forall\phi\in\DD(M)$, as is seen in Equation \ref{formequal}.
It generates a semigroup iff it is maximal accretive.\par
Suppose now that it is not maximal and that there exists an accretive extension $\bar{M}$,
being maximal and hence generating a semigroup.
In its domain $\DD(\bar{M})$, considered as a Hilbert space with norm (squared) $\|\phi\|_{\bar{M}}^2=\|\bar{M}\phi\|^2+\|\phi\|^2$,
there exists a non-empty subspace $\bar{\DD}$ orthogonal to $\DD(M)$.
Extend further the operators $L_k$ from acting on $\DD(M)$ to acting on  $\DD(\bar{M})$ by $L_k\phi=0$ for all $\phi\in\bar{\DD}$.
Then the inequality $\sum_k L_k^\dag L_k \leq\bar{M}+\bar{M}^\dag$ holds.
Therefore, under these conditions there exists a contractive semigroup on $\mathfrak{T}(\HH)$
with generator $\bar{\LL}$ such that
$$
 \bar{\LL}(|\phi\rangle\langle\psi|)=
\sum_k |L_k\phi\rangle\langle L_k\psi| -|\bar{M}\phi\rangle\langle\psi| - |\phi\rangle\langle \bar{M}\psi|,\qquad\forall (\phi,\psi)\subset\DD(\bar{M}),
$$
as proven in \cite[Section 3]{F99}. However, such an $\bar{\LL}$ would be an extension of $\LL$,
a contradiction to the maximal dissipativity of generators of semigroups.
\end{proof}

The properties of $M$ and of the $L_k$ are in precise correlation
to the structure of generators of QMS as defined in \cite{F99} and \cite{FR06}.


\subsection{Deriving decomposability from the standard form}\label{standardtoadc}

We  assume the existence of a generator $\L$ in generalized standard form
and we divide it, as above, into two pieces.
One of it is the completely positive map $\L_+$, the other one shall now be investigated in more detail.
\begin{define}\label{msuperdef}
On the Banach space $\Tf(\HH)$ we consider the operator
              \beq\label{msuper}
              \M:\,\,\rho\mapsto M\rho + \rho M^\dag,
              \eeq
with domain $\DD_0$, as defined in Definition \ref{gksldh}
as a subset of $\DD(\L)$, namely
\beq\label{dmdefine}
\DD_0 := \{\rho = \sum_{n=1}^N   |\phi_n \rangle\langle\psi_n|,\, \{\phi_n ,\psi_n\}\in\DD(M),\,N<\infty\}\label{rhodec}.
\eeq
\end{define}
It follows that also $\L_+=\L+\M$ is well defined on $\DD_0$.
Note that the expander $\L_+$ is dominated by $\M$. Namely, as in the proof of Proposition \ref{expander}, Equation (\ref{ltrace}), we have the general condition that
$\L$ generates a semigroup which preserves the trace, hence for all $\rho\in\DD(\L)$
\beq
\Tr\{\L(\rho)\}=\left[\frac d {dt}\Tr\{\Tt^t(\rho)\}\right]_{t=0}=0,
\eeq
implying
$$
0\leq \Tr\{\L_+(\rho)\}=\Tr\{\M(\rho)\}\leq \Tr\{|\M(\rho)|\} =\|\M(\rho)\|_1,\qquad \forall\rho\in\DD_0,\,\,\rho\geq 0.
$$
In order to take also non-positive matrices into account we look at general matrix units.
\begin{lem}\label{schwarztype}
For $(\phi,\psi)\subset\DD(M)$ we have
$$
\|\L_+( |\phi \rangle\langle\psi|)\|_1\leq  \left(\|\L_+( |\phi \rangle\langle\phi|)\|_1 \right)^{1/2}\cdot
\left(\|\L_+( |\psi \rangle\langle\psi|)\|_1 \right)^{1/2}.
$$
\end{lem}
\begin{proof}
We use generalized Kraus operators, see Section \ref{unboundedcp}. In case their number is infinite
we shall in the following first consider summations over finite $k\in\{0,\ldots, K\}$
and then in the end look at $K\rightarrow\infty$.
Moreover, the triangle inequality for the trace-norm
and the Schwarz inequality for the inner product in $\ell^2$ shall be used. We have
\beqa
\|\L_+( |\phi \rangle\langle\psi|)\|_1=  \|\sum_k |L_k \phi \rangle\langle L_k\psi| \|_1
\leq   \sum_k \| |L_k \phi \rangle\langle L_k\psi|\|_1 = \sum_k \|L_k \phi\|\cdot \| L_k\psi\| \nonumber\\
\leq  \left( \sum_k \|L_k \phi\|^2\right)^{1/2}\cdot\left(\sum_k \| L_k\psi\|^2 \right)^{1/2}
=\left(\|\L_+( |\phi \rangle\langle\phi|)\|_1 \right)^{1/2}\cdot
\left(\|\L_+( |\psi \rangle\langle\psi|)\|_1 \right)^{1/2}.  \nonumber
\eeqa
\end{proof}
For a general matrix unit we establish  approximations.
\begin{lem}\label{approx}
Consider two pairs,  $\{(\phi,\psi),\, (e,f)\}\subset\DD(M)$
which lie close to each other in the graph of $M$,
$$(\|\phi-e\|,\|M\phi-Me\|,\|\psi-f\|,\|M\psi-Mf\|)\subset[0,\eps].$$
Then for actions of super operators the following approximations hold
$$
\| \M( |\phi\rangle\langle\psi|)- \M( |e \rangle\langle f|)\|_1 \leq \eps\cdot (\|\phi\|+\|M\phi\|+\|\psi\|+\|M\psi\|+2\eps),
$$
and
$$
\| \L_+( |\phi\rangle\langle\psi|)- \L_+( |e \rangle\langle f|)\|_1 \leq \eps\cdot (\|\phi\|+\|M\phi\|+\|\psi\|+\|M\psi\|+2\eps).
$$
\end{lem}
\begin{proof}
\beqa
&&\| \M( |\phi\rangle\langle\psi|)- \M( |e \rangle\langle f|)\|_1 =
\| \M( |\phi-e\rangle\langle\psi|) + \M( |e \rangle\langle\psi- f|)\|_1 \nonumber\\
 &\leq& \|M\phi-M e\|\tdt\|\psi\| + \|\phi-e\|\tdt\|M\psi\| +\|Me\|\tdt\|\psi- f\|+\|e\|\tdt\|M\psi- M f\| \nonumber\\
&\leq& \eps\cdot (\|\psi\|+\|M\psi\|+\|\phi\|+\eps+\|M\phi\|+\eps).
\nonumber
\eeqa
For the action of $\L_+$ we use Lemma \ref{schwarztype} and domination by $\M$ when acting on positive $\rho$. Therefore, we have
\beqa
&&\| \L_+( |\phi\rangle\langle\psi|)- \L_+( |e \rangle\langle f|)\|_1 =       
\| \L_+( |\phi-e\rangle\langle\psi|) + \L_+( |e \rangle\langle\psi- f|)\|_1 \nonumber\\
&\leq&  \left(\|\L_+( |\phi -e\rangle\langle\phi-e|)\|_1 \right)^{1/2}\cdot
\left(\|\L_+( |\psi \rangle\langle\psi|)\|_1 \right)^{1/2} \nonumber\\
&& \qquad\qquad + \left(\|\L_+( |e\rangle\langle e|)\|_1 \right)^{1/2}\cdot
\left(\|\L_+( |\psi -f \rangle\langle\psi -f|)\|_1 \right)^{1/2}\nonumber\\
&\leq&  \left(\|\M( |\phi -e\rangle\langle\phi-e|)\|_1 \right)^{1/2}\cdot
\left(\|\M( |\psi \rangle\langle\psi|)\|_1 \right)^{1/2}     \nonumber\\
&& \qquad\qquad + \left(\|\M( |e\rangle\langle e|)\|_1 \right)^{1/2}\cdot
\left(\|\M( |\psi -f \rangle\langle\psi -f|)\|_1 \right)^{1/2}\nonumber\\
&\leq&  2\eps\cdot \left[ ( \|M\psi \|\tdt\|\psi\| )^{1/2}
 +  (\|M e\|\tdt\|e\| )^{1/2} \right]      \nonumber\\&\leq&
 \eps\cdot (\|M\psi\|+\|\psi\|+\|M\phi\|+\|\phi\|+2\eps).
\nonumber
\eeqa
Here we used the inequality between the geometric and arithmetic mean to transform the penultimate line.
\end{proof}

\begin{proposition}\label{easypart}
For a generator $\L$ in generalized standard form
there exists a basis $\{e_k\}$ such that $\DD_e$ is a core for $M$,
and $\DD_e^2$ is a core for $\L$.
\end{proposition}
\begin{proof}
The graph $\{(\phi,M\phi),\,\phi\in \DD(M)\}$
of the closed operator $M$
forms a complete Hilbert space with the graph norm $\|(\phi,M\phi)\|^2=\|\phi\|^2+\|M\phi\|^2$.
In this Hilbert space there exists a basis $\{(\phi_k,M\phi_k)\}$.
Since $\DD(M)$ is dense in $\HH$, the $\{\phi_k\}$ make a total set, and by the Gram Schmidt orthogonalization
one can construct out of them a basis $\{e_k\}$ of $\HH$.
Both the Gram Schmidt procedure and its reverse, representing each $\phi_k$
by the $\{e_k\}$, involve only finite linear combinations,
so the subspace $\DD_e$ is a core for $M$.

Now  a general $\rho\in\DD_0$ is of finite rank and thus a finite sum of matrix units,
$$
\rho=\sum_{n=1}^N |\phi_n\rangle\langle\psi_n|.
$$
For any $\eps>0$ there are vectors $e_n,\,f_n$ in $\DD_e$ close to $\phi_n,\,\psi_n$ in graph norm.
Due to these approximations in graph norm Lemma \ref{approx} can be applied. Combining both inequalities stated there, yields
$$
\| \L( |\phi_n\rangle\langle\psi_n|)- \L( |e_n \rangle\langle f_n|)\|_1 \leq 2\eps\cdot \left(\|\phi_n\|+\|M\phi_n\|+\|\psi_n\|+\|M\psi_n\|+2\eps\right).
$$
Then, by the triangle inequality and by summation of all these inequalities, we arrive at
\beqa
&&\|\L( \rho) - \L( \sum_n|e_n \rangle\langle f_n|)\|_1 \leq
\sum_n\| \L( |\phi_n\rangle\langle\psi_n|)- \L( |e_n \rangle\langle f_n|)\|_1 \nonumber\\
&&\leq N\tdt\eps\cdot (\max_n(\|\phi_n\|+\|M\phi_n\|+\|\psi_n\|+\|M\psi_n\|)+2\eps).
\eeqa

Now $\DD(\L)$ is the closure of the set of operators with finite rank in $\DD(\L)$,
so it follows that $\DD_e^2$ is a core for $\L$.
\end{proof}


\section{Discussion}\label{discussion}

QDS can be analyzed on three levels: the underlying Hilbert space, the space of Hilbert space operators,
and the ``super-operators'' acting on density operators.
On the ``highest level'', the level of super-operators, we
refer to general abstract theorems concerning strongly continuous semigroups
in Banach spaces.
The existence of a generator $\L$, under the condition that $\Tt^t$ is strongly continuous and trace preserving,
is the presupposition and basis for all our definitions, propositions and theorems.
On the intermediate level we discuss the form in which the generator can be represented
using operators on the Hilbert space. This problem is,
as is shown in this paper, related to the appearance of the QDS
on the lowest level.
On this level we may see, in principle, the evolution equation as a set of
coupled differential equations for
matrix elements in certain bases of $\HH$.

What is new in our presentation, in comparison to Davies' definition of a ``Standard Form''?
We use $M^\dag$ and $L^\dag$ at places where Davies puts - formally - the adjoint operators.
In the present context, the adjoint $M^\ast$ exists with a dense domain of definition, but its range is in general not large enough;
the operators $L_k$ are in general not closable, and the $L_k^\ast$  not densely defined.
These features can be explicitly seen in the examples illustrated in Section \ref{examples}.
Looking closely into Davies' papers, one may observe that he just writes  $M^\ast$ and $L^\ast$,
without really using these adjoint operators, when giving a precise sense to the formal notations.
Being honest, we admit that our ``generalized standard form'' is therefore not really new.
But his kind of notation may lead to some misunderstandings, which we try to avoid.
To characterize $M^\dag$ and $L^\dag$ more precisely,
one may note that these operators map $\HH$ into the dual space of $\HH_M$,
in the Gelfand triple
$\HH_M \subset \HH\subset \HH_M^\ast$,
where $\HH_M$ is
the domain of definition of $M$ equipped with the norm $\|\phi\|_M$, $\|\phi\|_M^2 = \|\phi\|^2+ \|M\phi\|^2$.

\textbf{Remark on the notation:}
      Our notation indicates that $\L(\rho)$ denotes  a trace class operator defined on the whole Hilbert space $\HH$.
      The restriction onto a domain concerns the set of density matrices, not of Hilbert space vectors.
      We choose nevertheless to write $\L$ in a way which is similar to Davies' expression of a standard form,
      and which looks like a formula in matrix analysis.
      This is a matter of taste; if one doesn't like the $M^\dag$ and the $L^\dag$,
      one can write, instead of (\ref{gensform}) and (\ref{formequal})
\beq\label{alternform}
\L(\rho) = \sum_k  L_k (L_k \rho^\ast)^\ast - M\tdt\rho - (M\tdt\rho^\ast)^\ast,
\eeq
with
\beq
\sum_k \| L_k \phi\|^2 =
\langle \phi | M \phi\rangle+\langle M \phi |\phi\rangle,\qquad\forall \phi\in\DD(M). \label{alterformequal}
\eeq

\textbf{Remark on the domains:} To construct the operators appearing in the generalized standard form
      it is not sufficient to assume that $\DD_e^2$ is contained in the domain of definition of $\L$.
      As a well-known counter example consider an induced unitary evolution $\Tt^t(\rho)=U_t\rho U_t^\ast$, where the generator
      of the group $U_t$ is a Laplacian on a finite interval $[a,b]$ with Dirichlet boundary conditions.
      With a basis $\{e_k\}$, where the limits $x\rightarrow a$ and  $x\rightarrow b$ of each function $e_k(x)$
      and of its derivative is $0$, the matrix elements give no information on boundary conditions. They define an
      hermitian operator with nonzero defect indices.

From the point of view of a functional analyst some important problems still demand further studies:
One is the question under which conditions a strongly continuous QDS is matrix normal
in an appropriate basis.
Another problem is how to reverse the implications, to
give a functional analytic characterization of the generators, appearing in the form
given in Definition \ref{gksldh},
without the a priori assumption that the semigroup is trace preserving.
This problem appears in the literature, starting essentially with \cite{C91}, continuing with
many more investigations, e.g. \cite{CF98},  \cite{FPM12},
as finding conditions for the generated semigroup to be ``conservative''.
Last but not least, one would like to have a ``standard non-standard form''
for generators which are not matrix normal.
We have examples, one from Davies and one from Holevo,
of such non-standard generators, \cite{D77,K95,H96}.
Both are built as extensions of generators in standard form,
but generating semigroups which do not preserve the trace.
The partition $\L=\L_+-\M$, defined here with Equations (\ref{gensform},\ref{formequal},\ref{expander},\ref{msuper})
is extended to $\L=\L_n+\L_+-\M$, including a non-standard completely positive map $\L_n$
which is not to be represented with generalized Kraus operators. Is this the ultimate most general form? A question that still needs to be answered.

\section{Examples}\label{examples}

\subsection{Translation, followed by drop out}\label{dropout}

The following simple example shows all the special features of the present analysis.
(A similar example is presented in \cite{A02}, sect. 2)
\par
Consider the Hilbert space
$$
\HH=\LL^2(\R_+)\oplus\C \,\, =\,\left\{ \left( \begin{array}{c}
                                                  \psi(x)\\
                                                   a
                                                \end{array}
\right)\right\}$$
and define $$|\omega\rangle :=  \left( \begin{array}{c}
                                                  0 \\
                                                  1
                                                \end{array}
\right),\quad \Omega := |\omega\rangle\langle\omega|.
$$
Moreover,
$$
\Tt^t(\rho)= \Tr(P_t \rho)\Omega + S_t \rho S_t^\ast
$$
with
$$
S_t|\omega\rangle=|\omega\rangle, \quad (S_t\psi)(x)=\psi(x+t),
\quad P_t |\omega\rangle=0,\quad  (P_t\psi)(x)=\chi_{[0,t]}(x)\psi(x).
$$
The generator is given by
$$\L(\rho)=L\rho\,L^\dag +\partial \rho+\rho\partial^\dag ,$$
where $\partial|\omega\rangle=0$ and $
(\partial\psi)(x)=\psi'(x)$ without boundary conditions, and
$$
 L         \left( \begin{array}{c}
                                                   \psi(x)\\
                                                  a
                                                \end{array}
\right) =  \left( \begin{array}{c}
                                                 0 \\
                                                 \psi(0)
                                                \end{array}
\right),\qquad \DD(L) \supset \DD(\partial)
$$
Note that $L$ -- formally $L=|\omega\rangle\langle\delta(x)|$ -- is not closable,
so $L^\ast$ is not a densely defined operator. As in Equation (\ref{genop1}), $L^\dag$   is a generalized operator,
a linear map from  $\HH$, the space of ket-vectors, to a space
of unbounded linear functionals on the Hilbert space of bra-vectors.
Formally $L^\dag=|\delta(x)\rangle\langle\omega|$.

The dual semigroup acts as
$$\Tt^{t\ast}(A)= \langle\omega|A|\omega\rangle  P_t +  S_t^\ast A S_t.
$$
It does not map $\CC (\HH)$, the space of compact operators, into $\CC (\HH)$. It, therefore, does not fulfill the condition
stated in Theorems 3.4 and 4.1 in \cite{D77b}, a requirement necessary in Davies' approach to prove
that the generator can appear in standard form.
It is not even strongly continuous and has no generator in the usual sense, namely
as an operator mapping $\BB(\HH)$ to $\BB(\HH)$.

One could define a generalized ``standard'' form for the dual generator,
mapping $\BB(\HH)$ to the set of quadratic forms $\mathfrak{Q}(\HH)$, by
$$
\L'(A)=L^\dag A\,L + \partial^\dag A+A\partial.
$$
It involves a completely positive mapping of $A$ to the quadratic form $L^\dag A\,L$.
A general theory of such extended Kraus operators is given in the Appendix
in Definition \ref{genops} and Theorem \ref{krausthm}.
Another characterization has been stated by Holevo in \cite{H97}.
He considered the generator as a bilinear functional, a ``form'', acting onto a pair
which consists of an observable and a density matrix.

We note that this example is not far from representing
a physical process as described in the following subsection.

\subsection{Quantum reflection with sticking probability}\label{sticking}

We present a simple model for a process observed in modern physics, see \cite{ZB83,BD06}.
Consider a hydrogen atom that moves freely in half space and is reflected from the surface of ultracold helium, the plane which is the boundary to the other half space. With some probability, however, it remains bounded to this surface
and creates there a ``ripplon''.\par
To describe this situation in a mathematical model, we consider the Fock space $\F$ of ripplons,
the space\footnote{
What we denote here by $\LL^2$ is the usual Hilbert space $L^2$, since the letter $L$ is reserved for operators.}
 $\LL^2(\R^2, dxdy)$ of wave functions for the hydrogen bounded to the surface,
and the space of wave functions in the half space, $\LL^2(\R^2, dxdy)\otimes \LL^2(\R_+,dz)$.
The total Hilbert space, therefore, is
\beqa
\HH&=& \F\otimes[\LL^2(\R^2, dx dy)\oplus \LL^2(\R^2, dx dy)\otimes \LL^2(\R_+,dz)]\nonumber\\
&=& \F\otimes \LL^2(\R^2, dx dy)\otimes (\C\oplus \LL^2(\R_+,dz)).
\eeqa
The free movement yields the Hamiltonian, here forming the imaginary part of the ``mover-and-contractor'' $M$,
$$
M=i(H_R +H_2 + H_z).
$$
$H_R$ is the self-adjoint Hamiltonian for the free ripplons.
$H_2$ is also self-adjoint and is the negative Laplacian on the plane, representing
both the free movement of the hydrogen atom when it is bounded to the surface,
and the parallel part of free movement in the half plane.
The evolution in the direction orthogonal to the boundary is modeled by $H_z=-\frac{d^2}{dz^2}$
with boundary conditions $\psi'(0)=w\psi(0)$, $w\in\C$, $\textrm{Im}(w)>0$.
This part is not self-adjoint, so $M+M^\dag=i(H_z-H_z^\dag)$. Moreover, the imaginary part of $w$ is related to the probability of sticking.
The transfer of the hydrogen atom into the layer at the surface is modeled by $L_z$ acting on
$\C\oplus\LL^2(\R_+,dz)$ as
$$
L_z|\psi\rangle=\left(2\,\textrm{Im}(w)\right)^{1/2}|\omega\rangle\psi(0),
$$
and $L_z|\omega\rangle=0$,
where $|\omega\rangle$ is a unit vector in the Hilbert space $\C$.
To complete the physics, there is an isometry from $\F\otimes \LL^2(\R^2)\otimes\C$ to $\F\otimes \LL^2(\R^2)\otimes\C$, namely the scattering matrix $S$.
It creates a ripplon with angular momentum $\kappa$
and reduces the momentum of the sticking hydrogen from $k$ to $k-\kappa$.
These vectors of momenta are all in the plane and $\kappa$ has the same direction as $k$. The absolute value of $\kappa$ depends on $|k|$ and is consistent
with the conservation of total energy.
On the total Hilbert space the generalized Kraus operator $L=S\tdt L_z$ is then acting.

From the perspective of mathematical physics, one would like to improve this model
and also to derive it from Schr\"{o}dinger quantum mechanics. Our conjecture is that this task can be fulfilled by using a strong coupling limit.

\section{Appendix: Generalized Kraus operators for unbounded completely positive maps}\label{unboundedcp}

\subsection{ Generalization of ``completely positive map''}

K. Kraus, \cite{K71} gave a representation of the Stinespring theorem
in the form
$$\Q(W)=\sum_k K_k\tdt W \tdt K_k^\ast,
$$
where $\Q$ is an ``operation'', a completely positive bounded map from $\mathfrak{T}(\HH)$ to $\mathfrak{T}(\HH)$.
It has been observed several times, e.g. in \cite{SG05}, that
in case of a finite dimensional Hilbert space $\HH$ this representation can be deduced from
the Choi-Jamiolkowski dualism of maps and states.

We give here a generalization to infinite dimensions and
densely defined unbounded completely positive maps.
Moreover, we consider the case, where the target of the map is not a space of operators
but the set of quadratic forms, $\mathfrak{Q}(\HH)$.

\begin{define}\textbf{Generalized Operators:}\label{genops}                    

Consider two bases, $\{e_n\}$ and $\{f_n\}$ in the Hilbert space of infinite dimension.
The operators $K$, used in the following, are linear mappings from $\DD_e$
into the algebraic dual of
$\DD_f$. Their adjoints, $K^\dag$, are operators from $\DD_f$
into the algebraic dual of $\DD_e$, related to $K$ by
\beq
\langle e|K^\dag|f\rangle= \langle f|K|e\rangle^\ast,
\eeq
for $e\in\DD_e$, $f\in\DD_f$.
\end{define}
In the sense of Holevo's investigations
these operators can also be considered as bilinear forms on $\lspan\{|f\rangle\langle e|\}$ and $\lspan\{|e\rangle\langle f|\}$ respectively.
In another sense one may characterize $M^\dag$ and $L^\dag$ more precisely
as acting  in the Gelfand triple
$\HH_M \subset \HH\subset \HH_M^\ast$.
These operators map there $\HH$ into the dual space of $\HH_M$, which is
the domain of definition of $M$ equipped with the norm $\|\phi\|_M$, $\|\phi\|_M^2 = \|\phi\|^2+ \|M\phi\|^2$.

\begin{thm}\textbf{Generalized Kraus-Operators}\label{krausthm}                  

If $\Q$ is a completely positive map from $\DD_e^2$ into $\mathfrak{Q}(\HH)$ with $\DD_f$ as common form domain,
there exists a countable set of ``generalized Kraus operators'' $K_\alpha$
such that
\beq
\Q(\rho) = \sum_\alpha  K_\alpha \rho K_\alpha^\dag.
\eeq
If, moreover,  each quadratic form $\Q(\rho)$ can be
associated with an operator $\breve{\Q}(\rho)$ in $\BB(\HH)$,
such that $\langle g|\Q(\rho)|f\rangle=\langle g|\breve{\Q}(\rho)f\rangle$ holds for all $f\in\DD_f$ and for all $g\in\DD_f$,             
each $K_\alpha$ is an operator mapping $\DD_e$ into $\HH$.
\end{thm}

Note that there exist moreover several considerations of ``Completely Positive Maps''
in a different context, see e.g. \cite{CTU11,P12}.


\subsection{Proof of Theorem \ref{krausthm} and Corollaries on special cases}

We consider matrix units $E_{i,j}=|e_i\rangle\langle e_j|$, so that
\beq\label{supermatrix}
Q_{i,j;k,\ell}=\langle f_i| \Q( E_{k,\ell})|f_j\rangle.
\eeq
Since the Choi matrix $E=\sum_{i,j}E_{i,j}\otimes E_{i,j}$
in infinite dimensions is a quadratic form \cite{H11}, not an operator,
we use the approximating increasing sequence of bounded positive operators on $\HH\otimes\HH$
$$E_N=\sum_{i,j}^N E_{i,j}\otimes E_{i,j}=|\theta_N\rangle\langle\theta_N|, \qquad
\mathrm{with} \qquad|\theta_N\rangle=\sum_i^N |e_i\otimes e_i\rangle.$$
Next, we make use of complete positivity of $\Q$ and look at the map $\hat{\Q}=\Q\otimes\one$, which preserves positivity. It is, therefore, possible to form the positive square root of $\hat{\Q}(E_N)$
and define vectors
$$\Psi_{N;i,k}=\sqrt{\hat{\Q}(E_N)}|f_i\otimes e_k\rangle.$$
These vectors are unnormalized and change with $N$. But observe that, for $N\geq k$
$$\|\Psi_{N;i,k}\|^2=\langle f_i\otimes e_k|\hat{\Q}(E_N)|f_i\otimes e_k\rangle=
\langle f_i|\Q(E_{k,k})|f_i\rangle \leq \infty$$
and, for $N\geq k$ and $N\geq \ell$

\beqa
\langle\Psi_{N;i,k}|\Psi_{N;j,\ell}\rangle=\langle f_i\otimes e_k|\hat{\Q}(E_N)|f_j\otimes e_\ell\rangle &=& \nonumber \\
=\langle f_i\otimes e_k|\Q(E_{k,\ell})\otimes E_{k,\ell} |f_j\otimes e_\ell\rangle &=& \langle f_i|\Q(E_{k,\ell})|f_j\rangle.     \label{matrixT}
\eeqa

The vectors $\Psi$ will, most probably, go weakly to zero, as $N\rightarrow\infty$,
but their norms and inner products are  bounded, increasing in $N$ and remaining constant for large $N$.
So the internal ``geometry'' of the rhomboid which they form stays constant and can be represented with a fixed set of vectors. This can be done by mimicking the Gram-Schmidt procedure, generating an orthonormal set of vectors out of the $\Psi$-vectors for any finite $N$,
and mapping them onto a fixed basis.

\begin{proposition} \textbf{Construction of the $N$-independent vectors} \label{conphi}

A set of vectors $\Phi_{i,k}\in\HH\otimes\HH$ can be constructed
independently of $N$, such that
\beq \label{psiphi}
\langle\Phi_{i,k}|\Phi_{j,\ell}\rangle = \langle \Psi_{N;i,k}|\Psi_{N;j,\ell}\rangle, \qquad
\forall N\geq \max\{i,j,k,\ell\}.
\eeq
\end{proposition}

\begin{proof}
We first choose the ordering of indices for the basis $e_k$ in such a way, that $\Psi_{N;0,0}\neq 0$ and then define some order relation for the index-pairs $(i,k)$.
More precisely, we set $n=n(i,k):=\max\{i,k\}$, $\alpha=\alpha(i,k):=n^2-n+i-k$ and
denote $\Phi_\alpha=\Phi_{i,k}$.

Consider now an arbitrary basis $b_\alpha$ and make the Ansatz:
$$\Phi_\alpha = \sum_{\beta \leq\alpha} \gamma_{\alpha,\beta}\,b_\beta.$$
Note that $\|\Phi_0\|^2 \neq 0$,
hence $$\gamma_{0,0}:=\|\Phi_0\|>0.$$

Next, we proceed with nested inductions on $\alpha$ and $\beta$ with $\beta\leq\alpha$.
Having determined all $\gamma_{\alpha-1,\beta}$  for fixed $\alpha-1$, we go over
to fixed $\alpha$ and determine the sequence $\gamma_{\alpha,\beta}$, in increasing order of $\beta$. To this end, we solve Equation (\ref{psiphi}) and take for $\beta<\alpha$ the expansion
$$\langle\Phi_\beta|\Phi_\alpha\rangle =
\left[\sum_{\beta' <\beta} \gamma_{\beta,\beta'}^\ast\,\gamma_{\alpha,\beta'}\right] +
\gamma_{\beta,\beta}^\ast\,\gamma_{\alpha,\beta},
$$
into account, which determines $\gamma_{\alpha,\beta}$ in case $\gamma_{\beta,\beta}\neq 0$.

If $\gamma_{\beta,\beta}= 0$, the vector $\Phi_\beta$ is a linear combination of $\Phi_{\beta'}$
with $\beta'<\beta$. This already guarantees the validity of Equation (\ref{psiphi}).
In such a case we reject any use of $b_\beta$ and set $\gamma_{\alpha,\beta}=0$ for each $\alpha>\beta$,
to make sure that the induction on $\beta$ can be  finished at $\beta=\alpha$.
This procedure guarantees, at the end of induction on $\beta$, the  solvability of
$$\langle\Phi_\alpha|\Phi_\alpha\rangle =
\left[\sum_{\beta<\alpha} |\gamma_{\alpha,\beta}|^2\right] +
|\gamma_{\alpha,\alpha}|^2,
$$
enabling us to determine $\gamma_{\alpha,\alpha}\geq 0$
as the positive square root of $\|\Phi_\alpha\|^2-\sum |\gamma_{\alpha,\beta}|^2$.
\end{proof}

One may represent now the super matrix elements by combining Equations (\ref{supermatrix}, \ref{matrixT}, \ref{psiphi})
\beq\label{represent}
Q_{i,j;k,\ell}
=\langle\Phi_{i,k}|\Phi_{j,\ell}\rangle
=\sum_\alpha\langle\Phi_{i,k}|\alpha\rangle\langle\alpha|\Phi_{j,\ell}\rangle,
\eeq
where $|\alpha\rangle$ are some basis vectors in $\HH\otimes\HH$.
We shall consider the basis $|\alpha\rangle =b_\alpha$ with $\alpha =\alpha(j,\ell)$,
which was used in the proof of the Proposition \ref{conphi}.
Note that, due to the Gram-Schmidt-like procedure used to construct the $\Phi$-vectors, there appears for each $\Phi_{i,k}$
only a finite set of $\alpha$ in Equation (\ref{defkraus}).
So one can swap the two sums in Equation (\ref{finite}) following below.

Define now ``Kraus-operators'' $K_\alpha$ on $\DD_e$  by their matrix-elements
\beq  \label{defkraus}
\langle f_i|K_\alpha|e_k\rangle := \langle \Phi_{i,k}|\alpha\rangle.
\eeq
Then with (\ref{represent}), we have for matrices $\rho$ which are finite linear combinations of $E_{k,\ell}$ that
\beqa\label{finite}
[\Q(\rho)]_{i,j}&=&\sum_{k,\ell}Q_{i,j;k,\ell}\,\rho_{k,\ell}=
\sum_{k,\ell} \sum_\alpha\langle\Phi_{i,k}|\alpha\rangle \rho_{k,\ell}   \langle\alpha|\Phi_{j,\ell}\rangle =
\sum_\alpha \sum_{k,\ell}\ldots =\\
&=&\sum_\alpha\left( K_\alpha \rho K_\alpha^\dag\right)_{i,j}.
\eeqa
%
%
%
%
\begin{cor}\label{krausdomain}
If $\Q$ maps $\DD_e^2$ into $\BB(\HH)$, the Kraus Operators $K_\alpha$ map $\DD_e$ into $\HH$.
\end{cor}
\begin{proof}
Consider $\phi\in\DD_e$ and form $|\phi\rangle\langle\phi|$ which is in $\DD_e^2$ and is mapped by $\Q$
to an operator with finite norm. For each $\psi\in\DD_f$ which is dense in $\HH$ we have
$$
\|\Q(|\phi\rangle\langle\phi|)\|\cdot \|\psi\|^2  \geq \langle\psi| \Q(|\phi\rangle\langle\phi|) |\psi\rangle =
\sum_\alpha \langle\psi| K_\alpha |\phi\rangle\langle\phi| K^\dag_\alpha |\psi\rangle
= \sum_\alpha | \langle\psi| K_\alpha |\phi\rangle|^2.
$$
So $\|\Q(|\phi\rangle\langle\phi|)\| $ gives a bound for the supremum over all normed $\psi$ and thus an upper bound for the square of
the norm of each $K_\alpha|\phi\rangle$.
\end{proof}
However, note that the $K_\alpha^\dag$ are in general still ``generalized'' operators.
\begin{cor}
For trace preserving maps $\Q$ the $K_\alpha$ can be extended to bounded operators in $\BB(\HH)$,
and $K_\alpha^\dag =K_\alpha^\ast$.
\end{cor}
\begin{proof}
For arbitrary positive $\rho$ the formula holds for
a sequence of approximating positive $\rho_N\in\DD_e^2$ with finite rank. Moreover, in the limit $N\to\infty$
we have bounded convergence.

Note that the preservation of trace by $\Q$ implies
$$\sum_i \|\Phi_{i,k}\|^2, = 1,\quad\forall\, k\qquad\Rightarrow\qquad
\sum_i \langle\Phi_{i,k}|\Phi_{i,\ell}\rangle = 0,\quad\forall\, k,\ell,$$
therefore,
$$\sum_\alpha K_\alpha^\dag \cdot K_\alpha =\one,\qquad\Rightarrow\qquad
\quad\|K_\alpha\|\leq 1,\quad\forall\,\alpha.
$$
\end{proof}

\subsection{Acknowledgements}

We thank Franco Fagnola, who spotted an error in the first version
and gave to us important hints on relations to existing papers.

One of us (S.A.) has been supported through the FWF-project Nr. P22929-N16.

B.B. thanks V. Umanit\`{a} and the organizing committee of the
Workshop ``Quantum Markov Semigroups:
 Decoherence and empirical estimates'',
 Genova, 26-28 June 2013,
for invitation to this marvelous workshop.
	
We thank also Jan Schlemmer for important hints on methods.


\end{document}